\documentclass[prx,aps,twocolumn,notitlepage,superscriptaddress,showpacs,showkeys,nofootinbib]{revtex4-2}

\usepackage{hyperref}
\usepackage{graphicx}
\usepackage{amsmath}
\usepackage{amssymb}
\usepackage{mathtools}
\usepackage{braket}
\usepackage{amsthm}
\usepackage{tabularx}
\usepackage{enumitem}

\usepackage{algorithm}
\usepackage{algcompatible}

\usepackage{comment}
\usepackage{array,multirow}

\usepackage{cleveref}
\crefname{algorithm}{Task}{Tasks}

\usepackage{tikz}
\usetikzlibrary{positioning, calc, shapes}

\usepackage{color}

\newtheorem{thm}{Theorem}

\newtheorem{lem}[thm]{Lemma}
\newtheorem{dfn}{Definition}
\newtheorem{remark}{Remark}

\newenvironment{conditions}
  {\par\vspace{\abovedisplayskip}\noindent\begin{tabular}{>{$}l<{$} @{${}:{}$} l}}
  {\end{tabular}\par\vspace{\belowdisplayskip}}

\DeclareMathOperator{\W}{W}
\DeclareMathOperator{\user}{User}
\DeclareMathOperator{\serv}{Serv}
\DeclareMathOperator{\key}{Key}
\DeclareMathOperator{\enc}{Enc}
\DeclareMathOperator{\dec}{Dec}

\begin{document}

\title{Quantum secure blind decryption with two users}
\author{Masahito Hayashi}\email{hmasahito@cuhk.edu.cn}
\affiliation{School of Data Science, The Chinese University of Hong Kong, Shenzhen, Longgang District, Shenzhen, 518172, China}
\affiliation{International Quantum Academy, Futian District, Shenzhen 518048, China}
\affiliation{Graduate School of Mathematics, Nagoya University, Nagoya, 464-8602, Japan}
\author{Yuki Ito}
\affiliation{Graduate School of Mathematics, Nagoya University, Nagoya, 464-8602, Japan}

\begin{abstract}
We propose two types of protocols for quantum secure blind decryption, involving two users and servers.
User 1 holds the encrypted ciphertext.
The servers store several indexed keys including 
the key encrypting the ciphertext.
User 2 aims to obtain the decrypted text.
The protocols are designed to preserve the following types of secrecy:
Users ensure the secrecy of the text from the servers.
Servers maintain the secrecy of the keys from the users.
Our protocols enable User 2 to obtain the decrypted text while preserving these secrecy requirements. 
Additionally, the second protocol ensures the secrecy of the key index to identify the key encrypting the ciphertext from the servers,
and the second protocol requires two non-commuting servers.
Furthermore, we analyze the secrecy of the second protocol under post-attack scenarios, where the two servers communicates with each other after the completion of the protocol.
We show that our quantum protocol satisfies the secrecy under these attacks, whereas its classical counterpart fails to do so.
\end{abstract}


\maketitle

\section{Introduction}
{\it Motivation:} 
A typical advantage of quantum information processing is its ability to provide various information-theoretic secure protocols, such as quantum key distribution and secure quantum computation.
This paper focuses on secure blind decryption, a communication task between users and servers.
%, as introduced in \cite{Chaum83, KY96, M11}.
In this task, the users possess an encrypted ciphertext, % using public-key cryptography, 
while the servers hold 
several indexed keys including the key encrypting the ciphertext.
The objective of secure blind decryption is to allow the users to decrypt their ciphertext through interaction with the servers without revealing its content to the servers. 
Additionally, the protocol ensures that the servers learn no information about the original message
as well as 
the key index to identify the key encrypting the ciphertext,
while one of the users %, starting with only the ciphertext, 
ultimately receives the decrypted message.

The preceding studies \cite{Chaum83, KY96, M11}
discussed a similar task with one user and one server.
Their protocols for this task rely on computational security, meaning their security depends on the computational hardness assumptions underlying public-key cryptography. 
However, such security becomes vulnerable if an efficient decryption algorithm or a quantum algorithm capable of breaking the cryptography is discovered. 
To mitigate this risk, it is necessary to develop a protocol that provides information-theoretic security for this task, particularly when the encryption employs Shannon's one-time pad keys.

In addition, their protocols work with only one key and 
ensure that the user obtains no information about the key
and the server obtains no information about the encrypted ciphertext
nor decrypted text.
But, their protocol does not cover the secrecy of the key index.
The goal of this paper is to propose quantum secure blind decryption protocols to fulfill the above secrecy requirements.
These protocols can be applied, for example, to securely managing wills containing sensitive information. 
In such a scenario, since parents do not wish the contents of their wills to be known,
they encrypt the will and leave the ciphertext with their lawyer, while the decryption key is sent to a trusted institution. 
Upon their passing, the secure decryption protocol ensures that the will can be securely transmitted from the lawyer to their child without revealing its contents to the institution.

{\it Overview of Quantum Secure Blind Decryption:}
The task of quantum secure blind decryption involves two users: $\user_1$, who possesses the ciphertext, and $\user_2$, to whom the decrypted message is intended to be sent via communication with the server. The ciphertext is encrypted using a one-time pad. Unlike public-key cryptography, the server maintains a database of decryption keys, and $\user_1$ holds a key index that specifies which key was used for encryption. The secrecy of the protocol is defined in terms of information-theoretic security.
We introduce two types of quantum secure blind decryption protocols, each with different secrecy requirements, and propose concrete protocols for each protocol.
In addition, we demonstrate that the proposed protocol for the second protocol achieves a performance unattainable by any classical system.

{\it First Result: Protocol for the First Setting:}
The first result introduces a new communication task and a corresponding protocol for the first setting, involving two users ($\user_1$ and $\user_2$) and one server. In this task, the two users are assumed to share entangled states beforehand. The purpose is to securely decrypt the ciphertext held by $\user_1$ and transmit the message to $\user_2$ without revealing the message to the server. While the message remains secret from the server, the user's key index is not.

Our concrete protocol for this setting relies on superdense coding, with the security analysis provided in \cite{DLL03,WLH22}. However, the first setting has potential risks: since the server knows which decryption key was used, it can reconstruct the message if $\user_1$ leaks the ciphertext after the protocol concludes.

{\it Second Result: Protocol for the Second Setting:}
The second result introduces a new task with the second setting, which provides stronger secrecy compared to the first setting. In this case, in addition to keeping the message secret, the protocol ensures that the key index also remains secure against the server. To achieve key index secrecy, the second setting assumes two servers that store identical databases. Unlike the first protocol, the second protocol does not require pre-shared entangled states.

{\it Third Result: Security against post attack:}
We emphasize that the second protocol achieves a level of secrecy unattainable in classical settings. In particular, we assume that the servers may communicate after the protocol concludes to infer the message, a scenario referred to as a post-attack model. 
This is because it is difficult to forbid the servers from communicating with each other after the completion of the protocol.
We prove that our quantum protocol preserves message secrecy under the post-attack model, whereas no classical protocol can achieve this, demonstrating a clear quantum advantage.

In contrast, the task in the first setting can be realized using classical methods. Therefore, our protocol in the second setting highlights the importance of quantum communication. While some may argue that presenting only the second setting would suffice, it is due to the complexity of the second setting we prefer to present the first setting beforehand.

\begin{widetext}
\begin{center}
\begin{table}[t]
\caption{Comparison with existing studies}
\label{T1}
\begin{center}
\begin{tabular}{|c|c|c|c|c|c|c|c|}
\hline
\multicolumn{2}{|c|}{} & \multirow{2}{*}{security} & \multirow{3}{*}{user} & \multirow{2}{*}{message} &\multirow{2}{*}{key}& key-& message  \\
\multicolumn{2}{|c|}{protocol} & \multirow{2}{*}{type} &  & \multirow{2}{*}{secrecy}&\multirow{2}{*}{secrecy} &index &secrecy against  \\
\multicolumn{2}{|c|}{}& &&&&secrecy& post-attack \\
\hline
\multicolumn{2}{|c|}{\multirow{2}{*}{\cite{Chaum83, KY96, M11}}}
& \multirow{2}{*}{CS}& one  & \multirow{2}{*}{Yes} & \multirow{2}{*}{Yes}& \multirow{2}{*}{No}&\multirow{2}{*}{N/A}\\
\multicolumn{2}{|c|}{} &&user&&&& \\
\hline
\multirow{2}{*}{one} &1st &\multirow{2}{*}{ITS} &two & \multirow{2}{*}{Yes} & \multirow{2}{*}{Yes}& \multirow{2}{*}{No} &\multirow{2}{*}{N/A}\\
&code & &users&&&& \\
\cline{2-8}
\multirow{2}{*}{server}&2nd &\multirow{2}{*}{ITS} &two & \multirow{2}{*}{Yes} & \multirow{2}{*}{No} & \multirow{2}{*}{Yes} &\multirow{2}{*}{N/A}\\
&code  & &users&&&& \\
\hline
\multirow{2}{*}{two} &\multirow{2}{*}{quantum}  &\multirow{2}{*}{ITS} &two & \multirow{2}{*}{Yes} & \multirow{2}{*}{Yes} & \multirow{2}{*}{Yes} & \multirow{2}{*}{Yes}\\
&& &users&&&& \\
\cline{2-8}
\multirow{2}{*}{servers}&\multirow{2}{*}{classical}  &\multirow{2}{*}{ITS} &two & \multirow{2}{*}{Yes} & \multirow{2}{*}{Yes} & \multirow{2}{*}{Yes} &\multirow{2}{*}{No}\\
&& &users&&&& \\
\hline
\end{tabular}
\end{center}
\vspace{2ex}
ITS means information-theoretic security.
CS means computational security.
Message secrecy means the message secrecy against the server(s).
Key secrecy means the key secrecy against the user(s).
Key-index secrecy means the key-index secrecy against the server(s).
\end{table}
\end{center}
\end{widetext}

{\it Organization of the Paper:}
The remainder of this paper is organized as follows:
Section \ref{S2} introduces the basic notation used throughout the paper.
Section \ref{S3} defines the first setting's task, quantum one-server protocol, and presents a concrete protocol along with its secrecy analysis.
Section \ref{S4} defines the second setting's task, quantum two-server protocol, and proposes a concrete protocol for this setting, with a discussion of its secrecy.
Section \ref{S5} analyzes secrecy under the post-attack model.
Section \ref{S6} shows that the classical case cannot achieves the performance presented in Section \ref{S6}.
Section \ref{S7} concludes the paper.

\section{Preliminary}\label{S2}
Before stating our protocols, we prepare fundamental knowledge for Bell states.
Let $\{\ket{0}, \ket{1}\}$ be an orthonormal basis of two-dimensional Hilbert space $\mathcal{A}$.
We define the Pauli operators $X, Z$ on $\mathcal{A}$ as
\begin{align}
    X = \ket{1}\!\bra{0} + \ket{0}\!\bra{1}, \quad
    Z = \ket{0}\!\bra{0} - \ket{1}\!\bra{1}.
\end{align}
These operators satisfy the following relation;
\begin{equation}\label{eq:commutative}
    XZ = - ZX.
\end{equation}
The maximally entangled state $\ket{\phi}$ on $\mathcal{A} \otimes \mathcal{A}$ is defined as
\begin{equation}
    \ket{\phi} = \frac{1}{\sqrt{2}}\left( \ket{00}+\ket{11}\right).
\end{equation}
For $a \in \mathbb{F}_2^n$, $i$-th element of $a$ is denoted as $a_i$,
\begin{equation*}
    a = (a_i,a_2,\dots,a_n).
\end{equation*}
We define the sum on $\mathbb{F}_2^n$ by the sum of each element on $\mathbb{F}_2$ as follows.
For $a, b \in \mathbb{F}_2^{n}$,
\begin{equation*}
    a \oplus b \coloneq (a_1 \oplus b_1, a_2 \oplus b_2, \dots a_n \oplus b_n).
\end{equation*}
For $k,j \in \mathbb{F}_2$, The discrete Weyl operator on 
the qubit system
$\mathbb{C}^2$ is defined as
\begin{equation}
    \operatorname{W}(k,j) \coloneq X^k Z^j.
\end{equation}
\if0
This operator satisfies the relations 
\begin{align}
    \W(k,j)^\dagger &= (-1)^{k \cdot j}\W(-k,-j), \\
    \W(k,j)\W(u,v) &= (-1)^{t \cdot u} \W(s \oplus u, t \oplus v).
\end{align}
\fi
For $s \in \mathbb{F}_2^{2n}$, 
the operator $\W_n(s)$ on 
$(\mathbb{C}^2)^{\otimes n}$ is defined by
\begin{equation}
    \operatorname{W}_n(s) \coloneq \W(s_1,s_2) \otimes \dots \otimes \W(s_{2n-1},s_{2n}).
\end{equation}

Next, for $k,j \in \mathbb{F}_2$,
we define the state $\ket{\phi_{kj}}$ on the composite system $\mathbb{C}^2 \otimes \mathbb{C}^2$ as
\begin{align*}
    \ket{\phi_{kj}} 
    &\coloneq \frac{1}{\sqrt{2}}(\ket{k}\!\ket{0} + (-1)^j \ket{k \oplus 1}\!\ket{1})\\
    &= (X^k \otimes Z^j) \ket{\phi} 
    = (\W(k,j) \otimes I)\ket{\phi}.
\end{align*}
\if0
Then, the set $\{\ket{\phi_{00}},\ket{\phi_{01}},\ket{\phi_{10}},\ket{\phi_{11}}\}$ forms 
an orthonormal basis of $\mathcal{A} \otimes \mathcal{A}$.
Further,
for $a,b,c,d,m,n \in \mathbb{F}_2$, the relation
    \begin{align}
&        (\W(a,b) \otimes \W(c,d)) 
\ket{\phi_{kj}}\!\bra{\phi_{kj}} (\W(a,b)^\dagger \otimes \W(c,d)^\dagger)\notag\\
            =& \ket{\phi_{k \oplus a \oplus c, j \oplus b \oplus d}}\!\bra{\phi_{k \oplus a \oplus c, j \oplus b \oplus d}}
    \end{align}
    holds.
\fi
In addition, throughout this paper, 
we use $2$ as the base of the logarithm.

\section{one-server protocol}\label{S3}
This section studies the communication task involving two users, $\user_1$ and $\user_2$, and a server.
The server stores the encryption keys, while $\user_1$ possesses the ciphertext of a message encrypted using one of these keys.
The goal of this task is to securely decrypt the ciphertext held by $\user_1$ and transmit the decrypted message to $\user_2$ without revealing the message to the server during the communication.
In addition, it is required that both users obtain no information for keys.
Notably, we do not impose any security constraints on the key index $K$, meaning the server is allowed to know the key index $K$.
We refer to this task as one-server protocol.

First, we formally define the protocol and its associated security constraints,
which formally clarifies our task.
Next, we present a concrete code for performing this task and analyze its security.

\subsection{Definition of protocol}
We present a formal description of secure blind decryption protocol without key index secrecy.
The communication flow is illustrated in Fig. \ref{fig: SDP without key index secrecy}.
The server stores $f$ keys, $\key_1, \dots, \key_f \in \mathbb{F}_2^{2n}$, which are uniformly and independently distributed.
$\user_1$ holds a key index $K$ and a ciphertext $E$ corresponding to a message $M \in \mathbb{F}_2^{2n}$, encrypted using the key $\key_K$, such that:
\begin{equation} E = M \oplus \key_K. \end{equation}
The task of this setting is the following.

Users $\user_1$ and $\user_2$ have access to $n$ qubit quantum systems, $\mathcal{H}^A_1, \dots, \mathcal{H}^A_n$ and $\mathcal{H}^B_1, \dots, \mathcal{H}^B_n$, respectively.
We define two $n$-qubit quantum systems as follows:
\begin{align} \mathcal{H}^A = \mathcal{H}^A_1 \otimes \dots \otimes \mathcal{H}^A_n, \quad \mathcal{H}^B = \mathcal{H}^B_1 \otimes \dots \otimes \mathcal{H}^B_n. \end{align}

%Users $\user_1$ and $\user_2$ share $n$ maximally entangled states, 
%$\ket{\phi} \in \mathcal{S}(\mathcal{H}^A_i \otimes \mathcal{H}^B_i)$, for $i = 1, \dots, n$.
Although the users do not communicate directly with each other, classical and quantum communications are permitted between the users and the server.
Then, our protocol is given as Protocol \ref{one-server-protocol}.
Detailed procedure of secure decryption protocol without key index secrecy depends on the 4-tuple $(\rho_{prev}, \enc_{user}, \enc_{serv}, \dec)$.
%$\rho_{prev} = \ket{\phi}^1 \otimes \dots \otimes \ket{\phi}^n$, 
We call it a code for the secure blind decryption without key index.

\begin{figure}[t]
    \centering
    \begin{tikzpicture}
        \node[draw,rounded corners](user1) {$\user_1$};
        \node[below right=3cm of user1.south,anchor=text,draw,rounded corners,rectangle split, rectangle split parts=2, align=center](serv) {$\serv$ \nodepart{two} $\key_1$ \\ ... \\$\key_f$};
        \node[above right=3cm of serv.text,draw,rounded corners](user2) {$\user_2$};

        \node[left=of user1](input) {$K, E$}; 
        \node[right=of user2](output) {$M$};
        
        \node[above=2.5cm of serv,draw](ent) {Shared entanglement $\rho_{prev}$};

        \draw[-latex] (user1) -- node[left=0.2]{$K, \mathcal{H}^{A'}$} (serv.text west);
        \draw[-latex] (serv.text east) -- node[right=0.2]{$\mathcal{H}^C$} (user2);

        \draw[double] (user1) -- (ent);
        \draw[double] (user2) -- (ent);

        \draw[dashed,-latex] (input) -- (user1);
        \draw[dashed,-latex] (user2) -- (output);
        
    \end{tikzpicture}
    \caption{Secure decryption protocol without key index secrecy.}
    \label{fig: SDP without key index secrecy}
\end{figure}
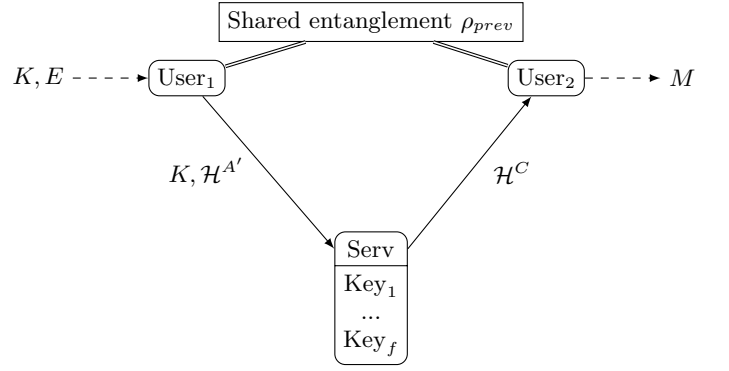

\begin{algorithm}[H]
\floatname{algorithm}{Protocol}
\caption{one-server protocol}
\label{one-server-protocol}
    \begin{description}
        \item[$\user_1$'s input] Key index $K \in \{1,\dots,f\}$, ciphertext $E = M \oplus \key_K \in \mathbb{F}_2^{2m}$.
        \item[$\user_2$'s output] Message $M \in \mathbb{F}_2^{2m}$.
    \end{description}

    \begin{enumerate}[label=\textbf{Step \arabic*.}]
        \setlength{\leftskip}{12mm}
        \item Before starting the protocol, $\user_1$ and $\user_2$ share an entangled state $\rho_{prev}$ over $\mathcal{H}^{A}\otimes\mathcal{H}^{B}$.
        \item Given an $n$-qubit quantum system $\mathcal{H}^{A'}$,
        $\user_1$ chooses the TPCP map $\Lambda_{user}$ from $\mathcal{H}^A$ to $\mathcal{H}^{A'}$ depending on the ciphertext $E$.
        This selection rule is denoted by the user encoder $\enc_{user}$;
        \begin{equation}
            \Lambda_{user} \coloneq \enc_{user}(E).
        \end{equation}
        $\user_1$ applies the map $\Lambda_{user}$ to $\mathcal{H}^A$ and sends the quantum system $\mathcal{H}^{A'}$ and the key index $K$ to the server.
    
        \item Given $n$-qubit quantum system $\mathcal{H}^C$, 
         the server chooses a TPCP map $\Lambda_{serv}$ from $\mathcal{H}^{A'}$ received from $\user_1$ to $\mathcal{H}^C$ depending on $K, \key_1,\dots \key_f$.
        This selection rule is denoted by the server encoder $\enc_{serv}$;
        \begin{equation}
            \Lambda_{serv} \coloneq \enc_{serv}(K, \key),
        \end{equation}
        where $\key = (\key_1, \dots, \key_f)$.
        The server applies the TP-CP map $\Lambda_{serv}$ to $\mathcal{H}^{A'}$,
and sends the quantum system $\mathcal{H}^C$ to $\user_2$.
        
        \item $\user_2$ performs the measurement defined by POVM $\dec$ on $\mathcal{H}^C \otimes \mathcal{H}^B$ to obtain the message $M$;
        \begin{equation}
            \dec = \{\Pi_\omega \mid \omega \in \mathbb{F}_2^{2n}\}.
        \end{equation}
    
        The output is the measurement outcome as a result of decryption.
    \end{enumerate}
\end{algorithm}

\subsection{Security conditions}
We consider the following types of security condition.
In the following conditions, we assume that all players make communication only at the time specified by the protocol.
Also, it is assumed that the massage $M$ is uniformly distributed.
\begin{description}[style=nextline]
    \item[Correctness]
    When the server and the users execute the protocol correctly, Protocol \ref{one-server-protocol} ensures that its output is the message $M$.
    In other words, we say that Protocol \ref{one-server-protocol} is correct when the error probability is zero, i.e., 
    \begin{equation}\label{eq:error_prob}
        P_e = \Pr[\omega \neq M] = 0.
    \end{equation}
    \item [Message-and-key secrecy against $\user_1$]
When $\user_1$ obtains no information for the message $M$ nor
the keys $\key$,
we say that Protocol \ref{one-server-protocol} 
satisfies the message-and-key secrecy against $\user_1$.
This condition always holds when
all players make communication only at the time specified by the protocol. Hence, we do not need to discuss this condition.

    \item[Key-and-key-index secrecy against $\user_2$]
    When $\user_2$ recovers the message $M$ correctly, 
$\user_2$ havs no information for the keys $\key$ stored by the server nor the key index $K$ held by $\user_1$.
    We formulate the key secrecy by the state that $\user_2$ received.
    Let $\rho_{user}(m, key,k)$ be a density matrix of the state that $\user_2$ received 
    when the message $M$ is $m$, the key index $K$ is $k$,
     and 
    the keys $\key$ are $key$. 
    We say that Protocol \ref{one-server-protocol} satisfies the 
    key-and-key-index secrecy against $\user_2$
when the condition
    \begin{equation}\label{eq:server_secrecy}
        \rho_{user}(m,key,k) = \rho_{user}(m,key',k') 
         \end{equation}
 holds for any $key, key',  k,  k'$.   

    \item [Message secrecy against the server]
    If $\user_1$ executes the protocol correctly according to Protocol \ref{one-server-protocol}, 
    the server obtains no information for the message $M$ regardless of the server's behavior.
        We formulate this secrecy as follows.
    When the message is $M$ and the key index is $K$, we denote the state the server receives by $\rho_{serv}(M, K)$.
    We say that Protocol \ref{one-server-protocol} satisfies the message secrecy if it satisfies the condition;
    \begin{equation}\label{eq:user_secrecy}
        \rho_{serv}(i,K) = \rho_{serv}(j,K) \quad \forall i,j \in \mathbb{F}_2^{2n}.
    \end{equation}

    \item[Key-index secrecy against the server]
When the users execute the protocol correctly, 
the servers obtain no information about the key index $K$ held by $\user_1$.
    The key index secrecy condition is defined by the independence between the key index $K$ an Queries $Q_1, Q_2$.
    We say that the code has key index secrecy when the condition     \begin{equation}
        I(K;Q_t) = 0 \quad (t = 1,2)
    \end{equation}
holds where $I(X;Y)$ is the mutual information between two random variables $X$ and $Y$.

\end{description}

\subsection{First type of code}\label{section:code_sdp_without_key_secrecy}
%\subsubsection{Code construction}
The first type of code is constructed as follows.
The initial state $\rho_{prev}$,
the user encoder $\enc_{user}$, the server encoder $\enc_{serv}$ and the decoder $\dec$ defined by
\begin{align*}
&m \coloneq n,\\
&\rho_{prev}\coloneq |\phi\rangle \langle \phi|^{\otimes n},\\
    &\enc_{user}(E) \coloneq  W_n(E), \\
    &\enc_{serv}(K, \key) \coloneq W_n(\key_K), \\
    &\dec 
 \coloneq \{\Pi_{\omega_1} \otimes \dots \otimes \Pi_{\omega_n} \mid \omega_i \in \mathbb{F}_2^2 \, (i = 1,\dots n) \},
\end{align*}
where $\{\Pi_{\omega_i}\}$ is a POVM of the basis measurement $\{\ket{\phi_{00}},\ket{\phi_{01}},\ket{\phi_{10}},\ket{\phi_{11}}\}$ on $\mathcal{H}_i^C \otimes \mathcal{H}_i^B$.

\begin{comment}
\begin{conditions}
    \key_1,\dots,\key_f & $\mathbb{F}_2^{2n}$の元でサーバが持つ鍵の一覧. \\
    M & $\mathbb{F}_2^{2n}$の元でオリジナルのメッセージ. \\
    K & $\{1,\dots,f\}$の元で,メッセージ$M$を暗号化した鍵のインデックス. \\
    E=M \oplus \key_K & メッセージ$M$を鍵$\key_K$によって暗号化したときの暗号化メッセージ. \\
    \mathcal{H}^A_i,\mathcal{H}^B_i & user1,user2が持っている1-qubit系.
\end{conditions}
\end{comment}

\begin{comment}
    
\begin{enumerate}[label=\textbf{Step \arabic*.}]
    \setlength{\leftskip}{6mm}
    \item $\user_1, \user_2$ share n two-dimensional maximally entangled states in advance.
    For $1 \le i \le n$, $user1$ applies $\W(E_i,E_{2i})$ to $\mathcal{H}^A_i$, respectively.
    
    \item $\user_1$ sends the quantum system $\mathcal{H}^A = \bigotimes_i \mathcal{H}^A_i$ and the key index $K$ via the quantum and classical communication channel.
    
    \item In this step, the server has the quantum system $\mathcal{H}^A$ and the key index $K$.
    The server prepares a map $\W_n(\key_K)$ from $\mathcal{H}^A$ to $\mathcal{H}^C$.
    Then the server applies the map to $\mathcal{H}^A$.
    Then the system $\mathcal{H}^C$ is sent to $\user_2$.

    \item $\user_2$ perform the basis measurement $\{\ket{\phi_{00}},\ket{\phi_{01}},\ket{\phi_{10}},\ket{\phi_{11}}\}$ on each $n$ composite systems $\mathcal{H}^C_i \otimes \mathcal{H}^B_i$ and obtain the outcome $(a_i,b_i)$, respectively.
    Finally, $\user_2$ obtain the output of the protocol $(a_1,\dots,a_n,b_1,\dots,b_n)$.
\end{enumerate}
\end{comment}

Then, we have the following theorem.
\begin{thm}\label{TH1}
The first type of code
presented in Section \ref{section:code_sdp_without_key_secrecy}
satisfies
the correctness, 
the message-and-key secrecy against $\user_1$,
the key-and-key-index secrecy against $\user_2$,
and the message secrecy against the server.
However, it does not satisfy the key-index secrecy against the server.
\end{thm}

\subsection{Proof of Theorem \ref{TH1}}
Since $\user_1$ sends the key index $K$ to the server,
the key-index secrecy against the server does not hold.
We show other condition as follows.

\subsubsection{Correctness}
We prove that $\user_2$ can obtain the correct message $M$ after the protocol is finished, assuming that the users and the server execute Protocol \ref{one-server-protocol} correctly. 
Let $\rho_{user,i}$ be the state that $\user_2$ receives on $i$-th composite system $\mathcal{H}^A_i \otimes \mathcal{H}^B_i$.
After Step 1, the state on $\mathcal{H}^A_i \otimes \mathcal{H}^B_i$ that $\user_1$ has is written as
\begin{equation}
    (\W(E_{2i-1},E_{2i}) \otimes I) \ket{\phi} = \ket{\phi_{E_{2i-1},E_{2i}}}.
\end{equation}
Since 
%Then, the state on $\mathcal{H}^A_i \otimes \mathcal{H}^B_i$ that $\user_2$ has is
\begin{align}
    &\W(\key_{K,2i-1},\key_{K,2i})\W(E_{2i-1},E_{2i}) \notag\\
    =& (-1)^{\key_{K,2i}\cdot E_{2i-1}}\W(M_{2i-1},M_{2i}),
\end{align}
the state $\rho_{user,i}$ is
\begin{align*}\label{eq:state1}
    \rho_{user,i} %\notag\\
        =& (\W(\key_{K,2i-1},\key_{K,2i}) \otimes I)
        \ket{\phi_{E_{2i-1},E_{2i}}} \!\notag\\ &\bra{\phi_{E_{2i-1},E_{2i}}}(\W(\key_{K,2i-1},\key_{K,2i}) \otimes I)^\dagger \\
        =& \ket{\phi_{M_{2i-1},M_{2i}}}\!\bra{\phi_{M_{2i-1},M_{2i}}}.
    \stepcounter{equation}\tag{\theequation}
\end{align*}
This implies that the measurement outcome of the basis measurement $\{\ket{\phi_{00}},\ket{\phi_{01}},\ket{\phi_{10}},\ket{\phi_{11}}\}$ on $\mathcal{H}^A_i \otimes \mathcal{H}^B_i$ is $(M_i,M_{2i})$ with probability 1, i.e., the relation
\begin{equation}
    \mathrm{Tr}\rho_{user,i} \ket{\phi_{k}}\!\bra{\phi_{k}}
        =
        \begin{cases}
            1 & k = (M_i, M_{2i}), \\
            0 & \textbf{otherwise}
        \end{cases}
\end{equation}
holds for  $k \in \mathbb{F}_2^2$.
By performing similar measurements on each composite system $\mathcal{H}^A_i \otimes \mathcal{H}^B_i$ respectively, $\user_2$ obtains the message $M = (M_1,M_2,\dots,\dots,M_{2n-1},M_{2n})$ with probability 1.

\subsubsection{Key-and-key-index secrecy against $\user_2$}\label{S3D3}
Assume that $\user_2$ recovers the message $M$ correctly.
    Let $\rho_{user}(m, key,k)$ be a density matrix of the state that $\user_2$ received 
    when the message $M$ is $m$, and the key index $K$ is $k$,
    the keys $\key$ are $key$. 
At the beginning of Step 3,
the chain rule of the mutual information guarantees that
    \begin{align}
I(M,K,\key;B,C)=    I(M;B,C)+ I(K,\key;B,C|M).
    \end{align}
    Since  $\user_2$ recovers the message $M$ correctly, 
    $I(M;B,C)\ge 2n$. Since the dimension of $\mathcal{H}^B\otimes \mathcal{H}^C$
   is $2^{2n}$,  $I(M,K,\key;B,C) \le 2n$.
   Thus, $I(K,\key;B,C|M)=0$, which implies
   the state $\rho_{user}(m, key,k)$ does not depend on $key,k$.
This fact shows that any code satisfy key secrecy in this protocol.
Further, $\user_2$ has no information 
for the key index $K$ as well as the keys $\key$.
In this derivation,
we assume only the no-communication condition between the users
for $\user_1$. 
That is, even when $\user_1$ does not follow the protocol with no communication with $\user_2$, the above analysis holds.

\subsubsection{Message secrecy against the server:}
Assuming that the users execute the protocol correctly, we prove that the message $M$ is secure from the server regardless of whether the server runs the protocol correctly.
In step. 3, the server obtains the maximally entangled state on $\mathcal{H}^A_i\,(i=1,\dots,n)$.
That is, the state on each $\mathcal{H}^A_i \, (i=1,\dots,n)$ obtained by the server through the protocol is  
\begin{equation}
    \rho_{serv,i} 
    = \mathrm{Tr}_B \ket{\phi_{E_i,E_{2i}}} \bra{\phi_{E_i,E_{2i}}}
    = \frac{1}{2} I.
\end{equation}
This implies that the entire state $\rho_{serv}(M,K)$ that the server receives is $\rho_{serv}(M, K) = \frac{1}{2^n}I$ and the relation
\begin{equation}
    \rho_{serv}(j,K) = \rho_{serv}(l,K) \quad \forall j,l \in \mathbb{F}_2^{2n}
\end{equation}
holds.
Therefore, we have proved the message secrecy defined by \eqref{eq:user_secrecy}.

\subsection{Second type of code}
The second type of code is constructed as follows.
We construct our code only when $n$ is an integer times of $f$. 
The initial state $\rho_{prev}$,
the user encoder $\enc_{user}$, and the server encoder $\enc_{serv}$ %and the decoder $\dec$ 
defined by
\begin{align*}
&m \coloneq  \frac{n}{f}-1 ,\\
&\rho_{prev}\coloneq |\phi\rangle \langle \phi|^{\otimes n},\\
    &\enc_{user}(E) \\
    &\coloneq 
    id^{\otimes (m+1)(K-1)} \otimes 
    W(1,0)\otimes W_m(E) \otimes id^{\otimes (m+1)(f-K)}, \\
    &\enc_{serv}(\key) \coloneq W_{m+1}(\key_1)\otimes \cdots \otimes, W_{m+1}(\key_f) 
\end{align*}
The decoder $\dec $ is determined as follows.
First, $\user_2$ applies the following measurement.
\begin{align*}
\{
 \Pi_{\omega_1} \otimes \dots \otimes \Pi_{\omega_n} \mid  \omega_i \in \mathbb{F}_2^2 \, (i = 1,\dots n) \},
\end{align*}
where $\{\Pi_{\omega_i}\}$ is a POVM of the basis measurement $\{\ket{\phi_{00}},\ket{\phi_{01}},\ket{\phi_{10}},\ket{\phi_{11}}\}$ on $\mathcal{H}_i^C \otimes \mathcal{H}_i^B$.
$\user_2$ finds an element $k$ such that
$\omega_{(m+1)(k-1)+1}=(1,0) $.
$\user_2$ sets the outcome to be
$\omega_{(m+1)(k-1)+2}, \ldots, \omega_{(m+1)k} $.

From the above construction, we can easily find the following lemma.
\begin{lem}
The second type of code satisfies
the correctness, 
the message-and-key secrecy against $\user_1$,
%the key-and-key-index secrecy against $\user_2$,
and the key-index secrecy against the server.
However, it does not satisfy 
the message secrecy against the server.
\end{lem}

\subsection{No go theorem}
\begin{thm}
No code for one-server protocol
satisfies
the correctness, 
the message secrecy against $\user_1$ and the server, 
the key secrecy against $\user_1$,
the key-and-key-index secrecy against $\user_2$, and
the key-index secrecy against the server.
\end{thm}

\begin{proof}
In order to prove this theorem by contradiction,
we assume that there exists 
a code for one-server protocol
that satisfies the above conditions.
We assume that 
$\user_1$ and $\user_2$ agree to make the following modification before the protocol.
$\user_1$ uses $(0,\ldots,0)$ instead of the message $M$.
$\user_2$ choose a value $K' \in \{1, \ldots, f\}$,
and asks $\user_1$ to use $K'$ as the key-index.
Then, after the protocol, $\user_2$ obtains $-\key_{K'}$ as the outcome.
Since all the conditions hold,
the above procedure realizes
symmetric private information retrieval between 
the server and $\user_2$.
However, since such a protocol does not exit \cite{Mayers,Lo},
we obtain the contradiction.
\end{proof}

\section{Two-server protocol}\label{S4}
In the previous section, Protocol \ref{one-server-protocol} permits the server to access the key index $K$ held by $\user_1$.
If the ciphertext $E$ is leaked for any reason, there is a security risk because the server could decrypt the leaked $E$ using the key $\key_K$ and retrieve the message $M$.
To mitigate this risk, we propose the concept of key index secrecy, which prevents a server from obtaining the message from the leaked encrypted ciphertext.
That is, we propose a new communication task,
% where the server cannot gain any information about the key index.
two-server protocol.
It is important to note that this task requires 
the secrecy of the key index in addition to other types of secrecy.

%If the encrypted ciphertext is leaked to the server under the scenario of Section \ref{S3}, the server can recover the message.
%To achieve this, we present our formulation incorporating key index secrecy.

\begin{comment}
** PIRと提案プロトコルの関係，相違点を書く．
** user,keyなどを数学記号として定義する
** 図をいれる
具体的には，2人のユーザuser1,user2と2つのサーバserv1,serv2があるとする．
前セクションと同様に，user1は暗号文$E = M \oplus \key_K$と鍵のインデックス$K$を持っており，serv1,serv2はすべての鍵$\key_1,\dots,\key_f$を管理している．
サーバ同士及びユーザ同士は通信できないが，ユーザとサーバ間では古典通信と量子通信ができるとする．
user1の目的は各サーバに対してクエリ$Q_1,Q_2 \subset \{1,\dots,f\}$と量子状態を送り，user2に復号化したメッセージ$M$を送ることである．
\end{comment}

\subsection{Definition of protocol}\label{section:def_sdp_with_key_secrecy}
The protocol for two-server protocol is defined as follows.
The communication flow and procedure for this protocol are described in Fig \ref{fig: SDP with key index secrecy}.
This protocol involves two users, $\user_1$ and $\user_2$, and two servers, 
$\serv_1$ and $\serv_2$.
The servers, $\serv_1$ and $\serv_2$, store $f$ keys, $\key_1, \dots, \key_f \in \mathbb{F}_2^{2n}$.
These keys are assumed to be uniformly and independently distributed.
$\user_1$ possesses a key index $K$ and a ciphertext $E$ encrypted using the key $\key_K$, such that $E = M \oplus \key_K$.
\if0
Additionally, $\user_1$ has access to $2n$ quantum systems, $\mathcal{H}^A_i, \mathcal{H}^B_i , (i = 1, \dots, n)$, and $n$ maximally entangled states $\ket{\phi}^i \in \mathcal{S}(\mathcal{H}^A_i \otimes \mathcal{H}^B_i) , (i = 1, \dots, n)$.
We introduce the following abbreviations for the composite systems:
$
\mathcal{H}^A \coloneq \mathcal{H}^A_1 \otimes \dots \otimes \mathcal{H}^A_n, \,
\mathcal{H}^B \coloneq \mathcal{H}^B_1 \otimes \dots \otimes \mathcal{H}^B_n
$.
\fi
Then, our protocol is given as Protocol \ref{protocol: SDP with key index secrecy}.

The detail procedure of 
Protocol \ref{protocol: SDP with key index secrecy}
is determined by 
%We define a protocol of two-server protocol as 
a 4-tuple 
$(\enc_{user}, \enc_{serv_1}, \enc_{serv_2}, \dec)$. 
Hence, we call this 4-tuple a code and denote it by $\Phi_{2n}$.

\begin{figure}[t]
    \centering
    \begin{tikzpicture}
        \node[draw,rounded corners](user1) {$\user_1$};
        \node[above right=2cm of user1,anchor=text,draw,rounded corners,rectangle split, rectangle split parts=2, align=center](serv1) {$\serv_1$ \nodepart{two} $\key_1$ \\ $\vdots$ \\$\key_f$};
        \node[below right=2cm of user1,anchor=text,draw,rounded corners,rectangle split, rectangle split parts=2, align=center](serv2) {$\serv_2$ \nodepart{two} $\key_1$ \\ $\vdots$ \\$\key_f$};
        \node[right=5cm of user1,draw,rounded corners](user2) {$\user_2$};

        \node[left=of user1](input) {$K, E$};
        \node[below=of user2](output) {$M$};

        \draw[-latex] (user1) -- node[above left]{$Q_1, \tilde{\mathcal{A}}^1$} (serv1.text west);
        \draw[-latex] (user1) -- node[below left]{$Q_2, \tilde{\mathcal{A}}^2$} (serv2.text west);
        \draw[-latex] (serv1.text east) -- node[above=0.1]{$\mathcal{A}^1$} (user2);
        \draw[-latex] (serv2.text east) -- node[above=0.1]{$\mathcal{A}^2$} (user2);

        \draw[dashed,-latex] (input) -- (user1);
        \draw[dashed,-latex] (user2) -- (output);
    \end{tikzpicture}
    \caption{two-server protocol.}
    \label{fig: SDP with key index secrecy}
\end{figure}
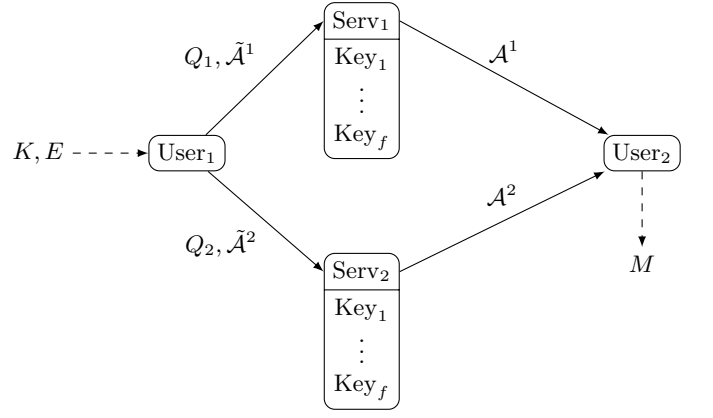

\begin{algorithm}[H]
\floatname{algorithm}{Protocol}
\caption{Quantum two-server protocol}
\label{protocol: SDP with key index secrecy}
    \begin{description}
        \item[$\user_1$'s input] Key index $K \in \{1,\dots,f\}$, ciphertext $E = M \oplus \key_K \in \mathbb{F}_2^{2m}$
        \item[$\user_2$'s output] Message $M \in \mathbb{F}_2^{2m}$
    \end{description}

    \begin{enumerate}[label=\textbf{Step \arabic*.}]
        \setlength{\leftskip}{12mm}
        \item Given $n$-qubit quantum systems $\tilde{\mathcal{A}}^1,\tilde{\mathcal{A}}^2$,
        $\user_1$ chooses a state $\rho_{E,K}$ on the composite system 
        $\tilde{\mathcal{A}}^1\otimes \tilde{\mathcal{A}}^2$,
%        the TPCP map $\Lambda_E$ from $\mathcal{H}^A \otimes \mathcal{H}^B$ to $\tilde{\mathcal{A}_1} \otimes \tilde{\mathcal{A}_2}$ 
        and randomly chooses two queries $Q_1, Q_2$.
These choices depend on the ciphertext $E$ and the key index $K$, and
        the selection rule is denoted by the user encoder $\enc_{user}$;
        \begin{equation}
            (\rho_{E,K}, (Q_1, Q_2)) \coloneq \enc_{user}(E,K,R).
        \end{equation}
    $\user_1$ sets the composite system $\tilde{\mathcal{A}}^1\otimes \tilde{\mathcal{A}}^2$ in the state $\rho_E$, and
sends the system $\tilde{A}^1$ and $\tilde{A}^2$ and queries $Q_1$ and $Q_2$ to $\serv_1$ and $\serv_2$, respectively.
        \item For $t=1,2$, given $n$-qubit quantum systems $\mathcal{A}^t$,
$\serv_t$ chooses the TPCP map $\Lambda'_t$ from $\tilde{\mathcal{A}}^t$ to $\mathcal{A}^t$ depending on $Q_t, \key \coloneq (\key_1,\dots,\key_f)$.
        This selection rule is denoted by the server encoder $\enc_{serv_t}$;
        \begin{equation}
            \Lambda'_t \coloneq \enc_{serv_t}(Q_t, \key).
        \end{equation}
        $\serv_t$ applies the TPCP map $\Lambda'_t$ from $\tilde{\mathcal{A}}^t$ to $\mathcal{A}^t$ depending on $Q_t, \key \coloneq (\key_1,\dots,\key_f)$.
        Then, $\serv_t$ sends the quantum system $\mathcal{A}^t$ to $\user_2$.
        \item $\user_2$ measures the received quantum system $\mathcal{A}^1 \otimes \mathcal{A}^2$ by the decoder $\dec$ defined by the POVM as 
        \begin{equation}
            \dec \coloneq \{\Pi_\omega \mid \omega \in \mathbb{F}_2^{2n}\}.
        \end{equation}
            A resultant of decryption is the measurement outcome.
    \end{enumerate}
    
\end{algorithm}

\subsection{Security condition}\label{S4B}
We consider the following types of security condition.
In the following conditions, we assume that all players make communication only at the time specified by the protocol.
Also, it is assumed that the massage $M$ is uniformly distributed.
The correctness, 
the message-and-key secrecy against $\user_1$,
and the key-and-key-index secrecy against $\user_2$
are defined in the same way as in 
Protocol \ref{one-server-protocol}.
The message secrecy against the servers
and the key-index secrecy against the servers
are defined as follows.

%We impose the following constraints to two-server protocol.

\begin{description}[style=nextline]
    \item[Message secrecy against the servers]
    If the users execute the protocol according to Protocol \ref{protocol: SDP with key index secrecy}, it is required that the servers can not obtain any information about the message $M$ even if they do not run the protocol correctly.
    Let $\rho_{serv_t}(M, K)$ be the state that $t$-th server $serv_t$ received when the message is $M$ and the key index is $K$.
    We say that the code satisfies the message secrecy if the condition
    \begin{equation}
        \rho_{serv_t}(i,K) = \rho_{serv_t}(j,K) \quad \forall i,j \in \mathbb{F}_2^{2n}
    \end{equation}
    holds.
\end{description}

\if0
Also, we define the quantities, the upload cost, the download cost, and the rate of a code $\Phi_{2n} = (\rho_{prev}, \enc_{user}, \enc_{serv}, \dec)$, by
\begin{align}
    U(\Phi_{2n}) &\coloneq \sum_{t=1}^2 \log_2 |Q_t|, \\
    D(\Phi_{2n}) &\coloneq \sum_{t=1}^2 \log_2 \dim \mathcal{A}^t, \\
    R(\Phi_{2n}) &\coloneq \frac{2n}{D(\Phi_{2n})}.
\end{align}
\fi

\begin{description}[style=nextline]
    \item[Key index secrecy against the servers]
    If the users execute the protocol correctly, it is necessary that the servers obtain no information about the key index $K$ held by $\user_1$.
    The key index secrecy condition is defined by the independence between the key index $K$ an Queries $Q_1, Q_2$.
    We say that the code has key index secrecy when the following condition holds;
    \begin{equation}
        I(K;Q_t) = 0 \quad (t = 1,2),
    \end{equation}
    where $I(X;Y)$ is the mutual information between two random variables $X$ and $Y$.
\end{description}

\subsection{Construction of our code}\label{section:code_sdp_with_key_secrecy}
We construct the concrete code for secure blind decryption protocol with key index secrecy.
Our code is inspired by the idea of 
%two servers PIR protocol given in \cite[Section 3]{CGKS98}.
quantum symmetric private information retrieval (QSPIR)
protocol with two servers given in \cite{SH21}.
%We assume that all the quantum systems are $2$-dimensional systems.
First, the state $\rho_{E,K}$ is defined as
\begin{equation}
\rho_{E,K}= \W_n(E)|\phi\rangle \langle \phi|^{\otimes n} \W_n(E)^\dagger.
\end{equation}
Then, let $R$ be a randomly chosen subset of $\{1, \dots, f\}$ and we prepare $Q_1$ and $Q_2$ as follows.
\begin{equation}\label{eq:query}
\begin{aligned}
    Q_1 &= R, \\
    Q_2 &=
      \begin{cases}
        Q_1 \setminus \{K\} & \textbf{if} \quad K \in Q_1 \\
        Q_1 \cup \{K\} & \textbf{otherwise}.
      \end{cases}
\end{aligned}
\end{equation}
The server encoders and the decoder is
\begin{align}
&    \enc_{serv_1}(Q_1, \key) \coloneq W_n(C^A), 
    \quad
    C^A = \bigoplus_{i \in Q_1} \key_i,
    \\
  &  \enc_{serv_2}(Q_2, \key) \coloneq W_n(C^B),
    \quad
    C^B = \bigoplus_{i \in Q_2} \key_i,
    \\
   & \dec \coloneq \{\Pi_{\omega_1} \otimes \dots \otimes \Pi_{\omega_n} \mid \omega_i \in \mathbb{F}_2^2 \, (i = 1,\dots n) \},
\end{align}
where $\{\Pi_{\omega_i}\}$ is a POVM of the basis measurement $\{\ket{\phi_{00}},\ket{\phi_{01}},\ket{\phi_{10}},\ket{\phi_{11}}\}$ on $\mathcal{H}_i^C \otimes \mathcal{H}_i^B$.

Then, we have the following theorem.
\begin{thm}\label{TH4}
The code presented in Section \ref{section:code_sdp_with_key_secrecy}
satisfies
the correctness, 
the message-and-key secrecy against $\user_1$,
the key-and-key-index secrecy against $\user_2$,
the message secrecy against the server,
and the key-index secrecy against the server.
\end{thm}

\subsection{Proof of Theorem \ref{TH4}}
\subsubsection{Correctness}
We prove that $\user_2$ can obtain the desired message $M$ when all the users and servers execute the protocol correctly.
At the step 1, 
$\user_1$ set the initial state on 
the $i$-th two-qubit system
$\tilde{\mathcal{A}}^1_i \otimes \tilde{\mathcal{A}}^2_i$ to be
the state $\ket{\phi_{{E_i},{E_{2i}}}}$.
Then, since the message $M$ has the relation $M = E \oplus \key_K$ and queries satisfy $C^A \oplus C^B = \key_K$, 
the state $\rho_{user,i}$ on the $i$-th two-qubit system
$\mathcal{A}^1_i \otimes \mathcal{A}^2_i$ at the beginning of Step 3 is described as
\begin{align*}\label{eq:state2}
    &\rho_{user,i}\notag\\
        =& (\W(C^A_i,C^A_{2i}) \otimes \W(C^B_i,C^B_{2i})) \notag\\
        &\ket{\phi_{E_i,E_{2i}}} \bra{\phi_{E_i,E_{2i}}} (\W(C^A_i,C^A_{2i})^\dagger \otimes \W(C^B_i,C^B_{2i})^\dagger) \\
        =& \ket{\phi_{E_i \oplus C^A_i \oplus C^B_i, E_{2i} \oplus C^A_{2i} \oplus C^B_{2i}}}\notag \\
        &\bra{\phi_{E_i \oplus C^A_i \oplus C^B_i, E_{2i} \oplus C^A_{2i} \oplus C^B_{2i}}} \\
        =& \ket{\phi_{M_i,M_{2i}}}\!\bra{\phi_{M_i,M_{2i}}}.
    \stepcounter{equation}\tag{\theequation}
\end{align*}
Therefore $\user_2$ obtains the message $M$ with probability 1 by the basis measurement $\{\ket{\phi_{00}},\ket{\phi_{01}},\ket{\phi_{10}},\ket{\phi_{11}}\}$ 
on $\mathcal{A}^1_i \otimes \mathcal{A}^2_i$, respectively.

\subsubsection{Key-and-key-index secrecy against $\user_2$}
Since the dimension of $\user_2$ receives is $2^{2n}$,
in the same way as Subsection \ref{S3D3}, we can show that 
any code of Protocol \ref{protocol: SDP with key index secrecy} satisfies 
the sever secrecy by using the chain rule of quantum mutual information.
 
\subsubsection{Message secrecy against the servers}
Assuming that the users follow the protocol, we prove that the message $M$ is secure against the servers even if they do not execute it correctly.
In Step 3, $t$-th server $\serv_t$ obtains only the system 
$\tilde{\mathcal{A}}^t$ which is one side of the composite system of the maximally entangled state.
Hence, when the message is $M$ and the key index is $K$, 
$\serv_t$ receives the following state $\rho_{serv_t}(M, K)$ in Step 2;
\begin{equation}
    \rho_{serv_t}(M,K) = \frac{1}{2^n} I.
\end{equation}
Therefore, the following relation holds for $\serv_t$ with $t=1,2$;
\begin{equation}
    \rho_{serv_t}(j,K) = \rho_{serv_t}(l,K) \quad \forall j,l \in \mathbb{F}_2^{2n},
\end{equation}
which shows the message secrecy against the servers.

\subsubsection{Key-index secrecy against the servers}
Assuming that the users run the protocol correctly, 
we prove that the servers obtains no information about the key index $K$.
Since the queries $Q_1$ and $Q_2$ are constructed randomly,
%from the set $\{1,\dots,f\}$, 
they are independent from $K$.
That is, the condition
\begin{equation*}
    I(K;Q_1) = I(K;Q_2) = 0
\end{equation*}
holds.
Therefore, we find that the code has the key-index secrecy against the servers.

\section{Secrecy under the post-specious-attack model}\label{S5}
In this section, we highlight the advantage of quantum two-server protocol over its classical counterpart under post-attack scenarios. Specifically, we define the concept of a post-attack and examine the classical version of
Protocol \ref{protocol: SDP with key index secrecy}.
Subsequently, we compare the classical and quantum versions in the context of post-attack situations.

\subsection{Definition of the post-specious-attack model}
In quantum blind decryption with key index secrecy, there are two potential risks that may arise after the completion of the protocol.
The first risk is that a server obtains the encrypted ciphertext.
The second risk is that the servers communicate with each other.
If both risks are realized, the servers can obtain the message because their communication enables them to identify the key index.
However, these risks occur independently, and thus the likelihood of both occurring simultaneously is relatively low.
It is therefore prudent to prepare for scenarios in which only one of these risks occurs.

We have already analyzed the secrecy of the message in the event of the first risk.
This section focuses on the second risk, i.e., the scenario where the two servers may communicate with each other after the protocol's completion if they act dishonestly.
In this case, deviations from the prescribed operations during the protocol could enable the servers to infer the message.
We now consider the scenario where the servers' behavior is specious \cite{DNS10}, which is explained as follows:
The servers' operations during the protocol may deviate from the correct operations, yet the information obtained by the users remains indistinguishable from the case where the servers strictly adhere to the protocol.
When the servers engage in specious behavior during the protocol and subsequently communicate with each other after its completion, we refer to this scenario as a post-specious-attack.
Since the discussion in the previous section does not address this type of attack, it is essential to analyze the message secrecy under post-specious-attacks.
To formalize this concept, we define a post-specious-attack as follows:

\begin{dfn}[Quantum post-specious-attack model]
The operations performed by the servers $\serv_1$ and $\serv_2$ are considered a post-specious-attack if they satisfy the following conditions:
\begin{itemize}
\item[PS1] The servers do not communicate with each other during the protocol.
\item[PS2] The servers communicate with each other after the protocol's completion.
\item[PS3] For $j = 1,2$, the server $\serv_j$ performs a local unitary operation $U_j$ on its local memory system ${\cal H}_{L(j)}$ and the received system ${\cal H}_j$, then sends the system ${\cal H}_j$ to $\user_2$.
Additionally, the initial state on ${\cal H}{L(j)}$ is a pure state $\rho_j$.
\item[PS4] The $\user_2$ correctly obtains the message $M \in \mathbb{F}_2^{2n}$ when both users act honestly. This condition is referred to as the specious condition.
\end{itemize}

We say that the protocol $\Phi_{2n}$ satisfies the message secrecy under the post-specious-attack if the servers $\serv_1$ and $\serv_2$ gain no information about the message $M$ from any post-specious-attack. Mathematically, this implies that the local memory systems ${\cal H}{L(1)}$ and ${\cal H}{L(2)}$ satisfy:
\begin{equation}
I(L(1),L(2),Q_1,Q_2,\key;M) = 0,
\end{equation}
after the protocol's completion for any post-specious-attack, assuming both users are honest.
\end{dfn}

The following points clarify the rationale behind this definition:
\begin{itemize}
\item Conditions PS1 and PS2: During the protocol, the servers are monitored by users and, therefore, cannot communicate with each other. However, after the protocol's completion, user monitoring ceases, making communication between servers plausible.
\item Condition PS3: By choosing the local memory system ${\cal H}{L(j)}$ to be sufficiently large, the local operation can always be expressed as a unitary operation $U_j$, with the initial state on ${\cal H}{L(j)}$ being a pure state $\rho_j$.
\item Condition PS4: To avoid detection, the servers’ attacks are restricted to specious attacks.
\end{itemize}

\subsection{Secrecy discussion}
Under this definition, we establish the following theorem:

\begin{thm}\label{thm: secrecy under post-attack model for quantum}
%The protocol for quantum secure blind decryption presented in Section \ref{section:code_sdp_with_key_secrecy} 
When a code for Protocol \ref{protocol: SDP with key index secrecy} satisfies all requirements given in Section \ref{S4B},
the code satisfies the message secrecy under the post-specious-attack model.
\end{thm}

\begin{proof}
Assume that both users are honest.
Now, we fix the variables 
$Q_1$, $Q_2$, $K$, $\key_{1},\ldots, \key_{f}$
to $q_1$, $q_2$, $k_0$, $k_{1},\ldots, k_{f}$.
%Choosing the local memory system ${\cal H}_{L(j)}$ largely,
%we can assume that $U_j$ is a unitary operation over ${\cal H}_{L(j)}\otimes {\cal H}_j$
%and the initial state $\rho_j$ on ${\cal H}_{L(j)}$ is a pure state,
%where ${\cal H}_j$ is a $2^n$-dimensional system.
Notice that the choice of $U_j$ depends on the $q_j$
and $k_{1},\ldots, k_{f}$ in Condition PS3.
Then, we define the channel $\Gamma_j$ as
\begin{align}
\Gamma_j(\rho):= \mathrm{Tr}_{L(j)} U_j (\rho\otimes \rho_j) U_j^\dagger.
\end{align}
Once $Q_1$, $Q_2$, $\key_{1},\ldots, \key_{f}$ are fixed
to $q_1$, $q_2$, $k_{1},\ldots, k_{f}$,
we denote the initial state with the encrypted text 
$e = m \oplus k_{k_0} \in \mathbb{F}_2^{2n}$
by $\tau_{e}$.
We denote the $\user_2$'s POVM 
by $\{\Pi_m\}_{m \in \mathbb{F}_2^{2n}}$.
Then, the specious condition guarantees that
\begin{align}
\mathrm{Tr} \Pi_{m'} (\Gamma_1\otimes \Gamma_2)(\tau_{m \oplus k_{k_0}})
=\delta_{m',m}
\end{align}
for $m',m \in \mathbb{F}_2^{2n}$.
Hence, we have
\begin{align}
\mathrm{Tr} 
(Id \otimes \Gamma_2^*)(\Pi_{m'}) 
(\Gamma_1\otimes Id)(\tau_{m \oplus k_{k_0}})
=\delta_{m',m}.
\end{align}
Then, 
$\{(\Gamma_1\otimes Id)(\tau_{m \oplus k_{k_0}})\}_{m \in \mathbb{F}_2^{2n}}$
are $2^{2n}$ orthogonal pure states.
Since $\tau_{m \oplus k_{k_0}}$ are maximally entangled states and
$\mathrm{Tr}_{1}(\Gamma_1\otimes Id)(\tau_{m \oplus k_{k_0}})
=\mathrm{Tr}_{1}(\tau_{m \oplus k_{k_0}})$,
$(\Gamma_1\otimes Id)(\tau_{m \oplus k_{k_0}})$
are also maximally entangled states.
Since the state 
$U_1 (\tau_{m \oplus k_{k_0}}\otimes \rho_1) U_1^\dagger$
is a pure state, 
the entropy of 
$(\Gamma_1\otimes Id)(\tau_{m \oplus k_{k_0}})
=\mathrm{Tr}_{L(1)}U_1 (\tau_{m \oplus k_{k_0}}\otimes \rho_1) U_1^\dagger$
equals 
the entropy of 
$\mathrm{Tr}_{1,2}U_1 (\tau_{m \oplus k_{k_0}}\otimes \rho_1) U_1^\dagger$.
Since $(\Gamma_1\otimes Id)(\tau_{m \oplus k_{k_0}})$ is a pure state,
$\mathrm{Tr}_{1,2}U_1 (\tau_{m \oplus k_{k_0}}\otimes \rho_1) U_1^\dagger$ is a pure state.
Further, the state 
$\mathrm{Tr}_{2} \tau_{m \oplus k_{k_0}}$ that 
$\serv_1$ receives does not depend on $m$.
Hence, the pure state
$\mathrm{Tr}_{1,2}U_1 (\tau_{m \oplus k_{k_0}}\otimes \rho_1) U_1^\dagger$ does not depend on $m$.
We denote it by $\kappa_1$.

We make the same discussion by exchanging the roles of 
$\serv_1$ and $\serv_2$.
Then, we find that 
$\mathrm{Tr}_{1,2}U_2 (\tau_{m \oplus k_{k_0}}\otimes \rho_2) U_2^\dagger$ is a pure state 
$\kappa_2$ that does not depend on $m$.
Therefore, after the completion of the protocol,
the state on the composite system 
${\cal H}_{L(1)}\otimes{\cal H}_{L(2)} $
is $\kappa_1\otimes\kappa_2$.
Therefore, the specious condition guarantees that
no post-specious-attack obtains the information for the message 
$M$ 
when both users are honest.
\end{proof}

\section{Classical version of secure blind decryption}\label{S6}
%and post-attack model}
\subsection{Classical version without key index secrecy}\label{S60}
The classical version of one-server protocol
can be implemented by using Shannon's one time pad key $R$ shared 
between two users as follows.
$\user_1$ sends the modulo sum $X=R \oplus E$
of the random number $R$ and the encrypted ciphertext $E= \key_K \oplus M$ and 
and the key index $K$ to the server.
The server calculates $Y= \key_K\oplus X$ and sends it to $\user_2$.
$\user_2$ obtains the message by $R \oplus Y=R \oplus \key_K\oplus R \oplus \key_K \oplus M=M $.

The correctness and the message secrecy are trivial.
The key secrecy can be shown as follows.
Assume that $\user_1$ and the server honest.
$\user_2$'s information is $M$ and $R$.
Since $R$ is independent of $\key$, the key secrecy holds.
However, the classical version with key index secrecy
is more complicated.

\subsection{Definition of classical protocol with key index secrecy}\label{S6A}
When the users and servers can only use classical computation and communication in the secure blind decryption protocol, 
we say that the protocol is classical. 
The formal definition of the classical secure blind decryption protocol with key index secrecy is shown below.  
$\user_1$ has the classical system $A,B = \mathbb{F}_2^n$ and the servers $\serv_1$ and $\serv_2$ possess the classical system $A',B' = \mathbb{F}_2^n$, respectively. 
The servers and the users are allowed to communicate with each other, but the servers cannot communicate with each other. 
A protocol for classical two-server protocol is defined as 
Protocol \ref{Protocol3}.

\begin{algorithm}[H]
\floatname{algorithm}{Protocol}
\caption{Classical version of secure decryption protocol with key index secrecy}
\label{Protocol3}
    \begin{description}
            \item[$\user_1$'s input] Key index $K \in \{1,\dots,f\}$, ciphertext $E = M \oplus \key_K \in A$
            \item[$\user_2$'s output] Message $M \in B$
    \end{description}
    \begin{enumerate}[label=\textbf{Step \arabic*.}]
        \setlength{\leftskip}{12mm}
    \item $\user_1$ generates random number $R_1$.
        Then, $\user_1$ prepares two strings $x_1,x_2 \in A$ based on $R_1$ and $E$.
        Also, $\user_1$ constructs two queries $Q_1,Q_2$ based on the key index $K$, ciphertext $E$, and $R_1$.
        All the operations are described by the user encoder $\enc_{\user}$:
        \begin{equation}
            (X_1, X_2, Q_1, Q_2) \coloneq \enc_{\user} (E,K,R_1).
        \end{equation}
        After preparation, $\user_1$ sends the classical systems $A$ and $B$ (that is $X_1,Q_1$ and $X_2,Q_2$) to the servers $\serv_1$ and $\serv_2$, respectively.
    \item The first(second) server applies the function $g_1 : A \to A$($g_2 : B \to B$) to received $X_t$.
        The function is constructed from the query $Q_t$ and $\key \coloneq (\key_1,\dots,\key_f)$, which is defined by server encoder $\enc_{\serv}$:
        \begin{equation}
            g_t \coloneq \enc_{\serv_t} (Q_t,\key).
        \end{equation}
        The $\serv_t$ sends the response $g_t(X_t)$ on the classical system $A$ to the $\user_2$.
    \item  The $\user_2$ reconstruct the message $M$ from the $g_1(X_1),g_2(X_2)$.
        This operation is written by the decoder $\dec$:
        \begin{equation}
            \dec(g_1(x_1), g_2(x_2)) = y.
        \end{equation}
        The result of the protocol is the decoder's output $y$.
    \end{enumerate}
\end{algorithm}

For this protocol, we define the correctness, 
the message secrecy, the sever secrecy, and the key index secrecy
in the same way as Protocol \ref{protocol: SDP with key index secrecy}.
The concrete form of Protocol \ref{Protocol3}
is determined by a $4$-tuple $(\enc_{\user},\enc_{\serv_1},\enc_{\serv_2}, \dec)$, which is called a code $\Phi_{n,c}$.

\subsection{Construction of our code with two servers}\label{section:code}
We consider the following code 
a $4$-tuple $(\enc_{\user},\enc_{\serv_1},\enc_{\serv_2}, \dec)$.
Assume that 
$R_1$ is composed of uniform random numbers
$R_{1,1}\in \mathbb{F}_2^n$ and $R_{1,2}\in \mathbb{F}_2^f$.
We define $(X_1, X_2, Q_1, Q_2) \coloneq \enc_{\user} (E,K,R_1)$ as
\begin{align}
X_1:= R_{1,1}, \quad X_2:= E\oplus  R_{1,1}, \quad Q_1:= R_{1,2}
\end{align}
and
\begin{align}
Q_{2,j}:=
\left\{
\begin{array}{ll}
 Q_{1,j} \oplus 1 & \hbox{ when } j=K \\
 Q_{1,j}  & \hbox{ when } j \neq K .
\end{array}
\right.
\end{align}
Then, we define
$g_t \coloneq \enc_{\serv_t} (Q_t,Key)$ 
and $\dec$ as
        \begin{align}
g_t(X_t) &:=\Big(\bigoplus_{j=1}^f Q_{t,j} \key_j\Big)\oplus  X_t \\
\dec(X_1', X_2')&:=X_1'\oplus X_2'.
        \end{align}

\subsection{Analysis of our code with key index secrecy}
We show the following theorem in this section.
\begin{thm}
The code presented in Section \ref{section:code}
satisfies the correctness, the message secrecy, the key secrecy
and the key index secrecy.
\end{thm}

\begin{proof}
When $\user_1,\user_2$ and $\serv_1$, $\serv_2$ are honest,
we have
\begin{align}
&\dec(g_1(X_1), g_2(X_2))\notag\\
=&
\Big(\Big(\bigoplus_{j=1}^f Q_{t,1} \key_j\Big)\oplus  X_1 
\Big)
\oplus
\Big(\Big(\bigoplus_{j=1}^f Q_{t,2} \key_j\Big)\oplus  X_2
\Big) \notag\\
=&
(\bigoplus_{j=1}^f (Q_{t,1}\oplus Q_{t,2}) \key_j\Big)\oplus  
(X_1\oplus  X_2)\notag\\
=&
 \key_K \oplus  E= M,
\end{align}
which shows the correctness.

Since $Q_t$ and $X_t$ are independent of $M$ and $\key$ when $\user_1$ is honest, the message secrecy and the key index secrecy are preserved.

When $\user_1$, $\serv_1$, and $\serv_2$ are honest, 
$\user_2$'s information consists of $g_1(X_1) , g_2(X_2) $.
That is, $\user_2$'s information is 
$f_1(X_1) \oplus f_2(X_2)=M$ and $f_1(X_1)$.
Since $R_{1,1}$ is an independent uniform random number,
$f_1(X_1)=\Big(\bigoplus_{j=1}^f Q_{2,j} \key_j\Big)\oplus  R_{1,1}
$ are independent of $\key$.
Hence, the key secrecy is maintained.
\end{proof}

\begin{remark}
The above derivation of the key secrecy
assume that $\user_1$ is honest.
If $\user_1$ behaves as follows,
$\user_2$ obtain $\key_K$ as follows.
Assume that
$\user_1$ fixes $R_{1,1}$ to be $0$, and 
decides $R_{1,2}$ as follows.
 \begin{align}
R_{1,2,j}:=
\left\{
\begin{array}{ll}
 1 & \hbox{ when } j=K \\
0  & \hbox{ when } j \neq K .
\end{array}
\right.
\end{align}
Then, $f_1(X_1)=\key_K$.
$\user_2$ obtains $\key_K$ as well as $M$.
 
However, as discussed in Subsection \ref{S3D3},
in the quantum case, the derivation of the key secrecy
assumes only the non-existence of the communication between the users.
Hence, quantum protocol for two-server protocol
guarantees the key secrecy with a weaker assumption than the above derivation.
It is not clear
whether there exists a code in the above classical protocol
such that 
the key secrecy only with 
the no-communication condition between users
holds in addition to the correctness, 
the message secrecy, and the key index secrecy.
\end{remark}

\subsection{Classical version of post-specious-attack model}
Next, we define the classical version of the post-specious-attack model. 
%Its definition is based on prior research \cite[B.1 Definition B.1]{DNS10}.
%In contrast to the prior research, which assumes a secure channel connecting the two servers, our task prohibits server communication during the protocol.
We model the classical post-attack as follows, using the same symbols as in 
Protocol 2:

\begin{dfn}[Classical post-attack model]\label{DF3}
The operations performed by the servers $\serv_1$ and $\serv_2$ are considered a post-attack if they satisfy the following conditions:
\begin{itemize}
\item[CP1] The servers follow the correct protocol during its execution.
\item[CP2] The servers communicate with each other after the protocol's completion.
\end{itemize}
We say that the protocol $\Phi_{n,c}$ maintains the message secrecy against post-attack if the servers $\serv_1$ and $\serv_2$ obtain no information about the message $M$ from any post-attack. Mathematically, this is expressed as:
\begin{equation}
I(X_1, X_2, Q_1, Q_2, \key; M) = 0.
\end{equation}
\end{dfn}

Under the above definition, we establish the following lemma:

\begin{lem}\label{lem: secrecy under the post-attack model for classical}
Any code $\Phi_{n,c} = (\enc_{\user}, \enc_{\serv_1},\enc_{\serv_2}, \dec)$ fails to satisfy the message secrecy under the post-attack model.
\end{lem}

\begin{proof}
For $t = 0,1$, $\serv_t$ obtains $g_t(X_t)$ using $X_t$ and
the server encoder $\enc_{\serv_t}(Q_t, \key)$.
Since the servers are allowed to communicate after the protocol is completed in the post-attack model, they can determine the message $M = \dec(g_1(X_1), g_2(X_2))$ using the decoder $\dec$.
Thus, no protocol $\Phi_{n,c}$ satisfies the message secrecy under the post-attack model.
\end{proof}

In the classical case, we do not use the term ``specious" because Definition \ref{DF3} prohibits the servers from deviating from the prescribed procedure.
As demonstrated in Lemma \ref{lem: secrecy under the post-attack model for classical}, even when the servers strictly follow the protocol, they can still succeed 
the post-attack.
Thus, if the servers are allowed to communicate with each other after the protocol's completion, they do not need to employ a specious attack to extract information.
This is why the concept of a post-specious-attack model is not introduced for classical protocols.

Theorem \ref{thm: secrecy under post-attack model for quantum} and Lemma \ref{lem: secrecy under the post-attack model for classical} collectively highlight the superiority of the quantum secure blind decryption protocol over its classical counterpart under the post-attack model.

\section{Conclusion}\label{S7}
We have proposed two types of new protocols for quantum secure blind decryption and constructed codes to realize these protocols.
This paper presents three main results:
The first and second results are the proposals of protocols for the two newly introduced quantum tasks,
quantum secure blind decryption without/with key index secrecy.
The first protocol ensures both user and key secrecy.
The second protocol extends this security by incorporating the key index secrecy.
We have formally defined these protocols and their associated secrecy requirements and have constructed concrete codes to achieve them.
As the third result, we have demonstrated that our second protocol satisfies secrecy against post-specious attacks, whereas its classical counterpart does not.
%Since post-specious attacks are most possible attacks,
%this advantage shows the usefulness of our proposed protocol for 
%two-server protocol.

We note that both codes for both protocols achieve the optimal transmission rate, as follows:
The upper bound of the transmission rate for the first protocol is derived by considering the entanglement-assisted channel \cite{BSST1,BSST2,Holevo}.
For the second protocol, the upper bound is determined by analyzing the channel between the two servers and $\user_2$.
In both cases, the transmission rates of our codes match these upper bounds.

%{\it Comparison with Symmetric PIR:}
Finally, we explain the relation with 
quantum symmetric private information retrieval (QSPIR), which is a similar task to quantum two-server protocol.
Here, we highlight the differences between the two tasks.
Private information retrieval (PIR) \cite{CGKS98, SJ17} allows a user to retrieve a file from a database stored by servers, such that the server learns no information about the file index. SPIR extends PIR by ensuring that the user gains no information about other files. Quantum extensions of SPIR, such as quantum symmetric PIR (QSPIR), have been studied in \cite{KdW04, Ole11, SH21}.

While the second setting, quantum two-server protocol
shares similarities with QSPIR, simply applying QSPIR to 
the transmission of the key used for the encryption
is insufficient due to the following reason.
In this case, $\user_2$ will obtain the key used for the encryption.
However, quantum two-server protocol
requires the secrecy of the key used for the encryption.
Hence, this application of QSPIR does not work for
quantum secure blind decryption with the index secrecy.

There are several future studies.

\end{document}